\documentclass{sig-alternate}
\usepackage{amsmath}
\usepackage{amsfonts}
\usepackage{amssymb}
\usepackage{graphicx}
\usepackage{subfigure}
\usepackage[outdir=./figures/]{epstopdf}
\usepackage{url}
\usepackage{color}
\usepackage{flushend}

\usepackage{algorithm}
\usepackage{algpseudocode}
\usepackage{enumerate}

\author{Heejin Ahn and Domitilla Del Vecchio}
\newtheorem{theorem}{Theorem}
\newtheorem{lemma}{Lemma}
\newtheorem{problem}{Problem}
\newdef{definition}{Definition}
\newdef{example}{Example}

\graphicspath{{./figures/}}
\begin{document}


\title{Semi-autonomous Intersection Collision Avoidance through Job-shop Scheduling}
%
%
%
%
%

\numberofauthors{2}

\author{
%
%
\alignauthor
Heejin Ahn\\
       \affaddr{Massachusetts Institute of Technology}\\
       \affaddr{77 Massachusetts Avenue}\\
       \affaddr{Cambridge, MA, 02139}\\
       \email{hjahn@mit.edu}
\alignauthor
Domitilla Del Vecchio\\
       \affaddr{Massachusetts Institute of Technology}\\
       \affaddr{77 Massachusetts Avenue}\\
       \affaddr{Cambridge, MA, 02139}\\
       \email{ddv@mit.edu}
}
\date{23 Oct 2015}

\maketitle
\begin{abstract}
In this paper, we design a supervisor to prevent vehicle collisions at intersections. An intersection is modeled as an area containing multiple conflict points where vehicle paths cross in the future. At every time step, the supervisor determines whether there will be more than one vehicle in the vicinity of a conflict point at the same time. If there is, then an impending collision is detected, and the supervisor overrides the drivers to avoid collision. 
A major challenge in the design of a supervisor as opposed to an autonomous vehicle controller is to verify whether future collisions will occur based on the current drivers choices. This verification problem is particularly hard due to the large number of vehicles often involved in intersection collision, to the multitude of conflict points, and to the vehicles dynamics. In order to solve the verification problem, we translate the problem to a job-shop scheduling problem that yields equivalent answers. The job-shop scheduling problem can, in turn, be transformed into a mixed-integer linear program when the vehicle dynamics are first-order dynamics, and can thus be solved by using a commercial solver. 
\end{abstract}




\section{Introduction}
\begin{figure}[t!]
\centering
\includegraphics[width = \columnwidth]{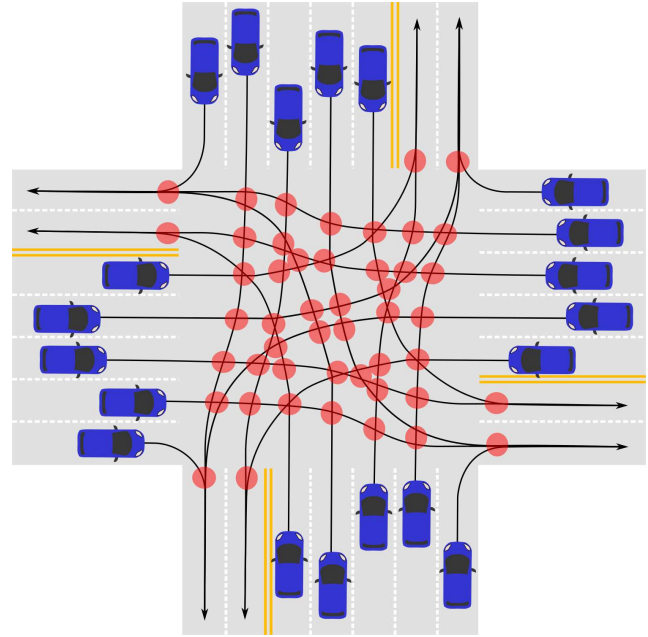}
\caption{General intersection scenario, taken from \cite{MassDOT_2012_Topcrash} to encompass the most dangerous intersections in Massachusetts, USA. This intersection contains forty eight conflict areas (small red circles). The supervisor designed in this paper can prevent collisions at the conflict areas by minimally overriding the vehicles.}
\label{figure:general_intersection}
\end{figure}

In the United States, 33,561 people lost their lives in vehicle crashes in 2012, and 26~\% of them occurred at or near intersections \cite{US_DOT_2012:overview}. This raises the need for improved safety systems that actively prevent collisions at intersections. For example, a centralized controller could be implemented on the infrastructure to coordinate vehicles near an intersection so as to prevent collisions. However, since a large number of vehicles are often involved in intersection collisions and vehicles are dynamic agents, the design of such systems faces challenges in terms of computational complexity. An additional substantial complication is that the system should override the drivers only when their driving will certainly cause a collision. That is, override actions should be minimally restrictive. This allows drivers to be in control of the vehicle unless unable to handle a dangerous situation. This supervisor can also be used as a safety guard for future fully autonomous vehicles driving in complex environment.

In this paper, we design a supervisor, which can be implemented on an infrastructure, communicating with human-driven vehicles near an intersection as shown in Figure~\ref{figure:general_intersection}. The most important and challenging part in the design is to determine whether vehicles' current driving will cause collisions at some future time. This is important because the exact collision detection, called the \textit{verification problem}, makes the supervisor least restrictive. This problem is not scalable with respect to the number of vehicles near an intersection yet their future safety must be verified every $\tau$ seconds, where $\tau$ is usually 100 ms \cite{US_plan_2015_2019}. To solve the verification problem in real-time, we formulate a job-shop scheduling problem, and prove that this is equivalent to the former problem. Although the job-shop scheduling problem is NP-hard \cite{garey_computers_1979}, we can solve this problem using a commercial solver by converting it into a mixed-integer linear programming problem. 

Mixed integer programming can handle both discrete and continuous aspects of a system. For example, collision avoidance can be formulated using discrete variables while the dynamic behaviors of vehicles, such as position and speed, are represented by continuous variables. Thus, mixed-integer programming has been employed in various collision avoidance applications such as air traffic control \cite{richards_spacecraft_2002,borrelli_milp_2006,christodoulou_automatic_2006} and multi-robot control \cite{earl_modeling_2002,grotli_path_2011}. Since the decision variables of these works are control inputs, for example, velocity, acceleration, or heading angle, at each time step within a finite time horizon, the discrete-time dynamics of vehicles are considered. As the number of time steps increases, the discretization error is diminished whereas the problem becomes larger and more difficult to solve. Because of this computational complexity, real-time verification is usually not feasible and hence not considered. Moreover, these works are cast in an autonomous framework in which if one input that satisfies the constraints is found, then it is applied. In contrast, in a semi-autonomous framework, such as ours, all admissible inputs need to be examined to determine if at least one feasible input exists.

In collision avoidance confined to an intersection, complexity can be mitigated by exploiting the fact that vehicles tend to follow predetermined paths. Given this, the intersection can be considered as a resource that all vehicles share. In \cite{kowshik_provable_2011, lee_development_2012,campos_cooperative_2014}, vehicles are assigned time slots during which they can be inside the intersection without conflict. Since the decision variables are the times at which each vehicle enters the intersection, the continuous dynamics are employed to compute these times. Notice that this approach considers $n$ decision variables if $n$ is the number of vehicles, whereas the approach in the previous paragraph considers at least $n*N$ decision variables if $N$ is the number of time steps on a finite time horizon. Because of the significantly smaller number of decision variables, the scheduling approach is computationally more efficient. The above works also assume full autonomy, which is not applicable to the scenarios considered in this paper. A detailed review of autonomous intersection management can be found in \cite{chen_cooperative_2015}.

\begin{figure}[tb!]
\centering
\includegraphics[width=\columnwidth]{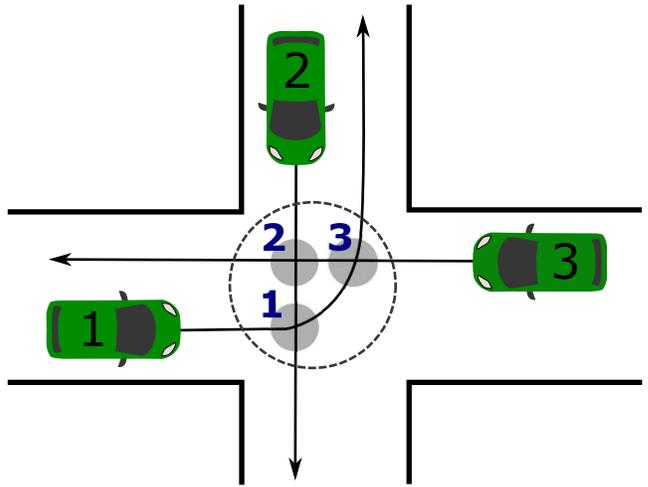}
\caption{Example of three vehicles with three conflict areas. The dashed circle represents the intersection model used in \cite{colombo_efficient_2012,colombo_least_2014}. In this paper, the intersection is modeled as multiple conflict areas as represented by the three shaded circles.}
\label{figure:scenario}
\end{figure}
A semi-autonomous framework with the scheduling approach is considered in \cite{colombo_efficient_2012,colombo_least_2014} by proving the equivalence between the verification problem and the scheduling problem. In these works, the authors design a least restrictive supervisor and restrict their attention to a special intersection scenario where all paths of vehicles intersect at one conflict area as indicated by the dashed region in Figure~\ref{figure:scenario}. While maintaining the same structure of the supervisor as in \cite{colombo_efficient_2012,colombo_least_2014}, we formulate a \textit{job-shop scheduling problem} to account for general scenarios of an intersection, where the paths of vehicles intersect at multiple points as in Figure~\ref{figure:general_intersection}. Considering multiple conflict points enables us to design a less conservative verification problem, but makes it more difficult to translate the problem to a job-shop scheduling problem. In this paper, we prove that our job-shop scheduling problem is equivalent to the verification problem with multiple conflict points. By virtue of this proof, we can solve the verification problem by solving the job-shop scheduling problem, which is computationally tractable. The job-shop scheduling problem is then transformed into a mixed-integer linear programming problem by assuming the single integrator dynamics of vehicles. Although a mixed-integer linear programming problem is NP-hard \cite{garey_computers_1979}, it can be solved by commercial solvers such as CPLEX \cite{Cplex_2009} or Gurobi \cite{Gurobi_2014}.
 
The rest of this paper is organized as follows. In Section~\ref{section:system_definition}, we introduce the intersection model and the dynamic model of vehicles. In Section~\ref{section:problem_statement}, we formally state the verification problem and the supervisor-design problem. The verification problem can be solved by formulating and solving a job-shop scheduling problem, which plays the most important role in the design of the supervisor. We then transform the job-shop scheduling problem into a mixed-integer linear programming problem to solve the job-shop scheduling problem using a commercial solver. These solutions will be given in Section~\ref{section:problem_solution}. We conclude this paper by presenting the results of computer simulations in Section~\ref{section:simulation_results} and conclusions in Section~\ref{section:conclusions}.

\section{System Definition}\label{section:system_definition}

Let us consider $n$ vehicles approaching an intersection. The vehicles follow their predetermined paths, and a point at which at least two of the paths intersect is defined as a \textit{conflict point}. Around a conflict point, we define a \textit{conflict area} to account for the size of vehicles. The intersection is modeled as a set of $m$ conflict areas as in Figures~\ref{figure:general_intersection} and \ref{figure:scenario}. Throughout this paper, vehicles and conflict areas are distinguished by integer indexes $1,\ldots, n$ and $1,\dots,m$, respectively. In order to focus only on intersection collision, we assume that there is only one vehicle per road. 

To model the longitudinal dynamics of vehicles, let $x_j\in X_j$ be the dynamic state of vehicle $j$. Let $u_j\in U_j\subset \mathbb{R}$ the control input of vehicle $j$. Then, the longitudinal dynamics are as follows:
\begin{align}\label{equation:dynamics}
\dot{x}_j=f_j(x_j,u_j), && y_j=h_j(x_j).
\end{align}
The output of the system is the position $y_j\in Y_j$ along the path. Here, $u_j$ is in a compact set, i.e., $u_j\in U_j:=[u_{j,min}, u_{j,max}]$. We assume that the output $y_j$ continuously depends on the input $u_j$. With abuse of notation, let $u_j$ denote the input signal as well as the input value in $\mathbb{R}$. The input signal $u_j\in\mathcal{U}_j$ is a function of time defined as $\{u_j(t): u_j(t)\in U_j~\text{for}~t\geq 0\}$.  

Let $x_j(t,u_j,x_j(0))$ denote the state reached after time $t$ with input signal $u_j$ starting from $x_j(0)$. Similarly, let $y_j(t,u_j,x_j(0))$ denote the position reached after time $t$ with input signal $u_j$ starting from $x_j(0)$. The aggregate state, output, input, and input signal are denoted by $\mathbf{x}\in\mathbf{X}, \mathbf{y}\in\mathbf{Y},\mathbf{u}\in\mathbf{U}$, and $\mathbf{u}\in\mathbf{\mathcal{U}}$, respectively.

One of the most important properties of the dynamic model \eqref{equation:dynamics} is the order-preserving property. That is, for $u_j(t)\leq u_j'(t)$ for all $t$, we have $x_j(t,u_j,x_j(0))\leq x_j(t,u_j',x_j(0))$ and $y_j(t,u_j,x_j(0))\leq y_j(t,u_j',x_j(0))$ for all $t\geq 0$. We will exploit this property in the design of the supervisor, particularly in formulating the job-shop scheduling problem.
	
\section{Problem Statement}\label{section:problem_statement}

Let $(\alpha_{ij}, \beta_{ij})\subset\mathbb{R}$ denote the location of conflict area $i$ along the longitudinal path of vehicle $j$. A conflict area is defined around a conflict point such that a collision occurs if more than one vehicle stay in a conflict area at the same time. That is, a collision occurs if $\mathbf{y}\in B$ where
\begin{align}\label{equation:badset}
\begin{split}
&B:=\{\mathbf{y}\in Y:~\text{for some}~j~\text{and}~j',\\
&\hspace{0.6 in} y_j\in (\alpha_{i,j}, \beta_{i,j})~\text{and}~y_{j'}\in (\alpha_{i,j'}, \beta_{i,j'})\}.
\end{split}
\end{align}
This subset of output $B$ is called the \textit{bad set}, and if $\mathbf{y}(t)\notin B$ for all $t\geq 0$, we consider the system \textit{safe}.

The verification problem is to determine if collisions at an intersection can be prevented at all future time given an initial state. We formally state this problem using the bad set \eqref{equation:badset} as follows.

\begin{problem}[Verification]\label{problem:verification}
Given $\mathbf{x}(0)$, determine if there exists $\mathbf{u}\in\mathcal{U}$ such that $\mathbf{y}(t,\mathbf{u},\mathbf{x}(0))\notin B$ for all $t\geq 0$.
\end{problem}

Now, we design a supervisor as follows. Every time $\tau$, the supervisor receives the measurements of current states of vehicles and drivers' inputs. Based on the measurements, the supervisor determines whether it must override the vehicles at this time step because otherwise there will be no admissible input to avoid collisions at the next time step. This decision can be made by solving the verification problem. 

The supervisor-design problem is formulated as follows.
\begin{problem}[Supervisor-design]\label{problem:supervisor}
	At time $k\tau$, given state $\mathbf{x}(k\tau)$ and drivers' input $\mathbf{u}_{driver}^k\in U$, design a supervisor that satisfies the following specifications.
	\begin{enumerate}[{Spec} 1.]
		\item For time $[k\tau, (k+1)\tau)$, it returns $\mathbf{u}_{driver}^k$ if there exists $\mathbf{u}\in\mathcal{U}$ such that for all $t\geq 0$ $$\mathbf{y}(t,\mathbf{u},\mathbf{x}(\tau, \mathbf{u}_{driver}^k,\mathbf{x}(k\tau)))\notin B,$$ or returns $\mathbf{u}_{safe}^k\in\mathcal{U}$ otherwise. Here, $\mathbf{u}_{safe}^k$ is defined as the safe input that guarantees the existence of $\mathbf{u}'\in\mathcal{U}$ such that for all $t\geq 0$, \begin{equation}\label{equation:safe_input_definition}
		\mathbf{y}(t,\mathbf{u}',\mathbf{x}(\tau, \mathbf{u}^k_{safe},\mathbf{x}(k\tau)))\notin B.
		\end{equation} \label{specification1}
		\item It is non-blocking, that is, $\mathbf{u}_{safe}^k$ must exist for any $k>0$ if $\mathbf{u}_{safe}^0$ exists.\label{specification2}
	\end{enumerate}
\end{problem}	

In Problem~\ref{problem:supervisor}, Spec~\ref{specification1} guarantees that the supervisor is least restrictive, and Spec~\ref{specification2} guarantees that the supervisor always has an input to override vehicles to ensure safety.

\section{Problem Solution}\label{section:problem_solution}
In this section, we solve the two problems: the verification problem (Problem~\ref{problem:verification}) and the supervisor-design problem (Problem~\ref{problem:supervisor}). As a main result, we formulate a job-shop scheduling problem and prove that this problem is equivalent to Problem~\ref{problem:verification}. Before formulating the job-shop scheduling problem in Section~\ref{section:verification_solution}, we introduce classical job-shop scheduling in Section~\ref{section:classical_job-shop}. In Section~\ref{section:verification_solution}, we also convert the job-shop scheduling problem into a mixed-integer linear programming problem with the assumption of first-order vehicle dynamics. In Section~\ref{section:supervisor}, the supervisor algorithm satisfying the specifications of Problem~\ref{problem:supervisor} is given.

\subsection{Classical job-shop scheduling}\label{section:classical_job-shop}
In classical job-shop scheduling \cite{pinedo_scheduling_2012}, $n$ jobs are processed on $m$ machines subject to the constraints that (a) each job has its own prescribed sequence of machines to follow, and (b) each machine can process at most one job at a time. 
This can be represented by a disjunctive graph with a set of nodes $\mathcal{N}$ and two sets of arcs $\mathcal{C}$ and $\mathcal{D}$. Here, the sets are defined as follows.
\begin{align*}
	&\mathcal{N}:=\{(i,j): (i,j)~\text{is the process of job $j$ on machine $i$}\\
	& \hspace{2 in}\text{for all }j\in\{1,\ldots,n\}\},\\
	&\mathcal{C}:=\{(i,j)\rightarrow (i',j): \text{job $j$ is must be processed on machine $i$}\\
	&\hspace{0.85 in}\text{and then on machine $i'$ for all }j\in\{1,\ldots,n\}\},\\
	&\mathcal{D}:=\{(i,j)\leftrightarrow (i,j'): \text{two jobs $j$ and $j'$ are to be processed}\\
	&\hspace{1.25 in}\text{on machine $i$ for all } i\in\{1,\ldots,m\}\}.
\end{align*}
The arcs in $\mathcal{C}$, called the conjunctive arcs, represent the routes of the jobs, and the arcs in $\mathcal{D}$, called the disjunctive arcs, connect two operations processed on a same machine. 

Let $\mathcal{F}\subseteq \mathcal{N}$ denote a set of the first operations of jobs, and $\mathcal{L}\subseteq \mathcal{N}$ denote a set of the last operations of jobs. If each job has only one operation on its route, $\mathcal{N}=\mathcal{F}=\mathcal{L}.$

The scenario in Section~\ref{section:system_definition} can be described in job-shop scheduling by considering vehicles as jobs and conflict areas as machines. For instance, each vehicle in Figure~\ref{figure:scenario} has its own prescribed route. Vehicle 1 crosses conflict area 1 first and then conflict area 3. At most one vehicle can be inside each conflict area at a time, because otherwise collisions occur. The corresponding disjunctive graph is shown in Figure~\ref{figure:disjunctiveGraph}.

\begin{figure}[t]
	\centering
	\includegraphics[width=0.6\columnwidth]{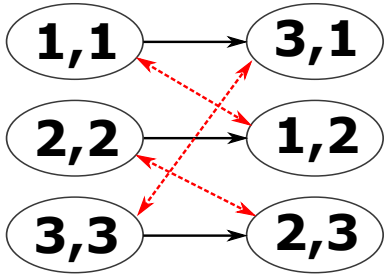}
	\caption{Disjunctive graph of the example in Figure~\ref{figure:scenario}. The black solid lines are the conjunctive arcs, and the red dotted lines are the disjunctive arcs.}
	\label{figure:disjunctiveGraph}
\end{figure}

\begin{example}
	The disjunctive graph in Figure~\ref{figure:disjunctiveGraph} consists of the set of nodes $$\mathcal{N} = \{(1,1),(3,1),(2,2),(1,2),(3,3),(2,3)\},$$ and the sets of conjunctive and disjunctive arcs\begin{align*}
		&\mathcal{C}=\{(1,1)\rightarrow(3,1), (2,2)\rightarrow(1,2), (3,3)\rightarrow(2,3)\},\\
		&\mathcal{D}=\{(1,1)\leftrightarrow(1,2),  (2,2)\leftrightarrow(2,3),(3,3)\leftrightarrow(3,1)\},
	\end{align*}
	respectively. 
	
	The sets of the first and the last operations are 	$\mathcal{F}=\{(1,1),(2,2),(3,3)\}$ and $\mathcal{L}=\{(3,1),(1,2),(2,3)\}$, respectively.
	
\end{example}

In \cite{balas_job_1998}, as a variant of job-shop scheduling, release times and deadlines are considered such that jobs must start after given release times and be finished before given deadlines. The release time $r_j$ and the deadline $d_j$ are defined for each job $j$, not for each operation $(i,j)$. The process time $p_j$ is a constant for all operations of job $j$ independent of the machines. Then, the classical job-shop scheduling problem with deadline is formulated as follows.
\begin{problem}[Classical job-shop]\label{problem:classical_jobshop}
	Given the release times $r_j$, the deadlines $d_j$, and the process time $p_j$, determine if there exists the operation starting times $t_{ij}$ for all $(i,j)\in\mathcal{N}$ such that
	\begin{align*}
		&\text{for all}~ (i,j)\in\mathcal{F}, && r_{j}\leq t_{ij},\\
		&\text{for all}~ (i,j)\in\mathcal{L}, && t_{ij}+p_{j}\leq d_{j},\\
		&\text{for all}~ (i,j)\rightarrow (i',j)\in\mathcal{C}, && t_{ij}+p_{j}\leq t_{i'j},\\
		&\text{for all}~ (i,j)\leftrightarrow (i,j')\in\mathcal{D}, && t_{ij}\leq t_{ij'}\Rightarrow t_{ij}+p_j\leq t_{ij'}.
	\end{align*}
\end{problem}

In the next section, a new job-shop scheduling problem similar to Problem~\ref{problem:classical_jobshop} is formulated to solve Problem~\ref{problem:verification}.

\subsection{Solution of Problem~\ref{problem:verification}}\label{section:verification_solution}

\subsubsection{Job-shop scheduling}\label{section:job-shop}
 In contrast to classical job-shop scheduling, our problem must account for the dynamic model of vehicles \eqref{equation:dynamics}. Thus, process times, release times, and deadlines are not initially given and not constant with operation starting times. Also, they are defined for each operation, that is, depending on the jobs and the machines as follows.

	\begin{figure}[tb!]
	\centering
	\includegraphics[width = \columnwidth]{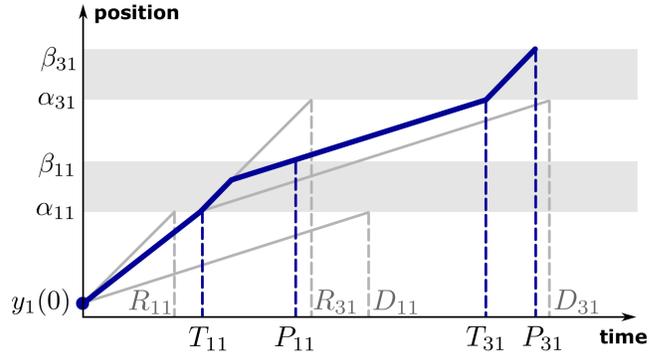}
	\caption{Process time, release time, and deadline of vehicle 1 in the example in Figure~\ref{figure:scenario}. The thick blue line represents the position of vehicle 1 on its longitudinal path with the schedules $T_{11}$ and $T_{31}$. For $(1,1)\notin\mathcal{L}$, $P_{11}$ is a function of $T_{11}$ and $T_{31}$, while for $(3,1)\in\mathcal{L}$, $P_{31}$ is a function of $T_{31}$ only. Also, for $(3,1)\notin\mathcal{F}$, $R_{31}$ and $D_{31}$ are functions of $T_{11}$.}
	\label{figure:RDP}
	\end{figure}

\begin{definition}\label{definition:job-shop}
	Given initial condition $\mathbf{x}(0)$ and schedule $\mathbf{T}:=\{T_{ij}\in\mathbb{R}: y_j(T_{ij},u_j,x_j(0))=\alpha_{ij}~\text{for some}~u_j\in\mathcal{U}_j,~\forall (i,j)\in\mathcal{N}\}$, process time $P_{ij}$ is defined for operation $(i,j)\in \mathcal{N}$ as follows.
	
	\begin{itemize}
		\item If $(i,j)\in\mathcal{L}$, for $y_j(0)<\alpha_{ij}$, 
		\begin{align}\label{definition:P_in_L}
		\begin{split}
		&P_{ij}:=\min_{u_j\in\mathcal{U}_j}\{t:y_j(t,u_j,x_j(0))=\beta_{ij}\\
		&\hspace{0.4 in}\text{with constraint}~y_j(T_{ij},u_j,x_j(0))=\alpha_{ij}\}.
		\end{split}
		\end{align}
		For $\alpha_{ij}\leq y_j(0)<\beta_{ij}$, set $P_{ij}:=\min_{u_j}\{t:y_j(t,u_j,x_j(0))=\beta_{ij}\}$. For $\beta_{ij}\leq y_j(0)$, set $P_{ij}=0$. If the constraint is not satisfied, set $P_{ij}=\infty$.
		\item If $(i,j)\notin \mathcal{L}$, that is, $\exists (i',j)$ such that $(i,j)\rightarrow(i',j)\in\mathcal{C}$, for $y_j(0)<\alpha_{ij}$,
		\begin{align}\label{definition:P_notin_L}
		\begin{split}
		&P_{ij}:=\min_{u_j\in\mathcal{U}_j}\{t:y_j(t,u_j,x_j(0))=\beta_{ij}\\
		&\hspace{0.4 in}\text{with constraint}~y_j(T_{ij},u_j,x_j(0))=\alpha_{ij}\\
		&\hspace{1.05 in}\text{and}~y_j(T_{i'j},u_j,x_j(0))=\alpha_{i'j}\}.
		\end{split}
		\end{align}
		For $\alpha_{ij}\leq y_j(0)<\beta_{ij}$, set $P_{ij}:=\min_{u_j}\{t:y_j(t,u_j,y_j(0))=\beta_{ij}~\text{with constraint}~y_j(T_{i'j},u_j,x_j(0))=\alpha_{i'j}\}$. For $\beta_{ij}\leq y_j(0)$, set $P_{ij}=0$. If the constraints are not satisfied, set $P_{ij}=\infty$.
	\end{itemize}
\end{definition}
By the above definition, process time $P_{ij}$ is the earliest time at which vehicle $j$ can exit conflict area $i$.

\begin{definition}
	Given initial condition $\mathbf{x}(0)$ and schedule $\mathbf{T}:=\{T_{ij}\in\mathbb{R}: y_j(T_{ij},u_j,x_j(0))=\alpha_{ij}~\text{for some}~u_j\in\mathcal{U}_j,~\forall (i,j)\in\mathcal{N}\}$, release time $R_{ij}$ and deadline $D_{ij}$ are defined for operation $(i,j)\in \mathcal{N}$ as follows.
	
	\begin{itemize}
		\item If $(i,j)\in\mathcal{F}$, for $y_{j}(0)<\alpha_{ij}$, 
		\begin{align}\label{definition:RD_in_F}
		\begin{split}
		&R_{ij}:=\min_{u_j\in\mathcal{U}_j}\{t:y_j(t,u_j,x_j(0))=\alpha_{ij}\},\\
		&D_{ij}:=\max_{u_j\in\mathcal{U}_j}\{t:y_j(t,u_j,x_j(0))=\alpha_{ij}\}.
		\end{split}
		\end{align}
		For $\alpha_{ij}\leq y_j(0)$, set $R_{ij}=0$ and $D_{ij}=0$.
		\item If $(i,j)\notin\mathcal{F}$, that is, $\exists(i',j)$ such that $(i',j)\rightarrow (i,j)\in\mathcal{C}$, for $y_j(0)<\alpha_{ij}$,
		\begin{align}\label{definition:RD_notin_F}
		\begin{split}
		&R_{ij}:=\min_{u_j\in\mathcal{U}_j}\{t:y_j(t,u_j,x_j(0))=\alpha_{ij}\\
		&\hspace{0.4in}\text{with constraint}~y_j(T_{i'j},u_j,x_j(0))=\alpha_{i'j}\},\\
		&D_{ij}:=\max_{u_j\in\mathcal{U}_j}\{t:y_j(t,u_j,x_j(0))=\alpha_{ij}\\
		&\hspace{0.4in}\text{with constraint}~y_j(T_{i'j},u_j,x_j(0))=\alpha_{i'j}\}.
		\end{split}
		\end{align}
		For $\alpha_{ij}\leq y_j(0)$, set $R_{ij}=0$ and $D_{ij}=0$. If the constraint cannot be satisfied by any $u_j\in\mathcal{U}_j$, set $R_{ij}=\infty$ and $D_{ij}=-\infty$.
	\end{itemize}	
	\end{definition}
	By definition, release time $R_{ij}$ is the earliest time at which vehicle $j$ can enter conflict area $i$, and deadline $D_{ij}$ is the latest such time.

	If an intersection is modeled as a single conflict point as in \cite{colombo_efficient_2012,colombo_least_2014}, the process time is defined by \eqref{definition:P_in_L}, and the release time and deadline by \eqref{definition:RD_in_F}. This is because each vehicle has a single operation so that $\mathcal{F}=\mathcal{L}=\mathcal{N}$. As for multiple conflict points, we have to include the effect of preceding and succeeding operations in the definition. Notice that the process time $P_{ij}$ in \eqref{definition:P_notin_L} depends on the schedules $T_{ij}$ and $T_{i'j}$ where $(i',j)$ is the succeeding operation of $(i,j)$, and the release time $R_{ij}$ and deadline $D_{ij}$ in \eqref{definition:RD_notin_F} depend on $T_{i'j}$ where $(i,j)$ is the preceding operation of $(i,j)$. An example of these definitions is illustrated in Figure~\ref{figure:RDP}.

Using the above definitions, we formulate the job-shop scheduling problem as follows.
\begin{problem}[Job-shop Scheduling]\label{problem:job-shop}
Given $\mathbf{x}(0)$, determine the existence of a schedule $\mathbf{T}:=\{T_{ij}:(i,j)\in\mathcal{N}\}\in\mathbb{R}_+^{|\mathcal{N}|}$ that satisfies the following constraints.
\begin{align}
&\text{for all}~(i,j)\in\mathcal{N}, && R_{ij}\leq T_{ij}\leq D_{ij},\label{constraint1:physical}\\
&\text{for all}~(i,j)\leftrightarrow (i,j') \in\mathcal{D}, && T_{ij}\leq T_{ij'}\Rightarrow P_{ij} \leq T_{ij'}.\label{constraint2:BI}
\end{align}
\end{problem}

Constraint~\eqref{constraint2:BI} implies avoidance of intersection collisions between vehicles $j$ and $j'$ by ensuring that vehicle $j$ must exit conflict area $i$ before vehicle $j'$ enters it.

We now prove that Problem~\ref{problem:verification} is equivalent to Problem~\ref{problem:job-shop}. Before this, we introduce the formal definition of the equivalence between two problems, and prove a lemma that relates Constraint~\eqref{constraint1:physical} to the existence of an input $u_j$ such that $y_j(T_{ij},u_j,x_j(0))=\alpha_{ij}$ and $y_j(P_{ij},u_j,x_j(0))=\beta_{ij}$ for $(i,j)\in\mathcal{N}$.

\begin{definition}\cite{cormen_introduction_2009}
	An instance $I_A$ of Problem A is the information required to solve the problem. If $I_A$ satisfies Problem A, we write $I_A\in$~Problem A. 
	
	Problem A is reducible to Problem B if for any instance $I_A$ of Problem A, an instance $I_B$ of Problem B can be constructed in polynomial time, and $I_A\in$~Problem A if and only if $I_B\in$~Problem B. If Problem A is reducible to Problem B, and Problem B is reducible to Problem A, then Problem A is equivalent to Problem B.
\end{definition}
\begin{lemma}\label{lemma:existence_input}
	If $R_{ij}\leq T_{ij}\leq D_{ij}$ for all $(i,j)\in\mathcal{N}$ with $y_j(0)<\alpha_{ij}$, there exists $u_j\in\mathcal{U}_j$ such that $y_j(T_{ij},u_j,x_j(0))=\alpha_{ij}$ and $y_j(P_{ij},u_j,x_j(0))=\beta_{ij}$.
\end{lemma}
\begin{proof}
	By the definitions of $R_{ij}$ and $D_{ij}$ in \eqref{definition:RD_in_F}, for the first operation $(i_1,j)$ on the route of vehicle $j$, there exists an input signal ${u}_j$ such that $y_j({T}_{i_1j},{u}_j,x_j(0)) = \alpha_{i_1j}$. This is because the input space is path-connected, and the output $y_j$ continuously depends on $u_j$. Then, for the next operation $(i_2,j)$, that is, $(i_1,j)\rightarrow (i_2,j)\in\mathcal{C}$, since the constraint in definition \eqref{definition:RD_notin_F} is satisfied by the input signal $u_j$, there is an input signal $u_j$ such that $y_j(T_{i_2j},u_j,x_j(0))=\alpha_{i_2j}$. By induction on the sequence of operations, for all $(i,j)\in\mathcal{N}$, there exists an input signal $u_j$ such that $y_j(T_{ij},u_j,y_j(0))=\alpha_{ij}$.
	
	This input signal $u_j$ satisfies the constraints in the definition of $P_{ij}$ in \eqref{definition:P_in_L} and \eqref{definition:P_notin_L}. Since there exists at least one input signal that satisfies the constraints, an input signal $u_j$ exists such that $y_j(P_{ij},u_j,x_j(0))=\beta_{ij}$ for all $(i,j)\in\mathcal{N}$.
\end{proof}
\begin{theorem}\label{theorem:P1_equivalent_P3}
Problem~\ref{problem:verification} is equivalent to Problem~\ref{problem:job-shop}.
\end{theorem}
\begin{proof}
An instance of Problem~\ref{problem:verification} is $\{\mathbf{x}(0),\Theta\}$, where  $\Theta=\{\{\alpha_{ij}, \beta_{ij}: \forall (i,j)\in\mathcal{N}\}, d, X,Y, U, \mathcal{U, N,F,L,C,D}\}$. 
Notice that an instance of Problem~\ref{problem:job-shop} is $\{\mathbf{x}(0),\Theta\}$, which is identical to an instance of Problem~\ref{problem:verification}. Thus, the construction of an instance takes $O(1)$ time. All we have to show is that given $I=\{\mathbf{x}(0),\Theta\}$, $I\in$~Problem~\ref{problem:verification} if and only if $I\in$~Problem~\ref{problem:job-shop}.

Suppose $I\in$~Problem~\ref{problem:verification}. Then, there exists $\tilde{\mathbf{u}}\in \mathcal{U}$ such that $\mathbf{y}(t,\tilde{\mathbf{u}},\mathbf{x}(0))\notin B$ for all $t\geq 0$. In this proof, we assume $y_j(0)<\alpha_{ij}$. For all $(i,j)\in\mathcal{N}$, let $\tilde{T}_{ij}=\{t:y_j(t,\tilde{u}_j,x_j(0))=\alpha_{ij}\}$ and $\tilde{P}_{ij} = \{t:y_j(t,\tilde{u}_j,x_j(0))=\beta_{ij}\}$. We will show that $\{\tilde{T}_{ij}:(i,j)\in\mathcal{N}\}$ satisfies the constraints in Problem~\ref{problem:job-shop} so that $\{\mathbf{x}(0),\Theta\}\in$~Problem~\ref{problem:job-shop}.

By the definitions of $R_{ij}$ and $D_{ij}$, we have $R_{ij}\leq \tilde{T}_{ij}\leq D_{ij}$ (Constraint~\eqref{constraint1:physical}). For all $(i,j)\leftrightarrow (i,j')\in\mathcal{D}$, assume without loss of generality vehicle $j$ enters conflict area $i$ before vehicle $j'$. Then we know that at $t=\tilde{P}_{ij}$, since $y_j(t,\tilde{u}_j,x_j(0))=\beta_{ij}$, we have $y_{j'}(t,\tilde{u}_j,x_j(0))\leq \alpha_{ij'}$. That is, $\tilde{P}_{ij} \leq \tilde{T}_{ij'}$. Since $\tilde{u}_j$ satisfies all the constraints given in the definitions of $P_{ij}$, we have $P_{ij}\leq \tilde{P}_{ij}$. Therefore, $P_{ij}\leq \tilde{T}_{ij'}$ (Constraint~\eqref{constraint2:BI}).

Suppose $I\in$~Problem~\ref{problem:job-shop}. Then, there exists $\hat{\mathbf{T}}$ satisfying the constraints in Problem~\ref{problem:job-shop}. By Lemma~\ref{lemma:existence_input}, there exists $\hat{\mathbf{u}}$ that satisfies $y_j(\hat{T}_{ij},\hat{u}_j,x_j(0))=\alpha_{ij}$ and $ y_j(P_{ij},\hat{u}_j,x_j(0))=\beta_{ij}$ for all $(i,j)\in\mathcal{N}$. In Constraint~\eqref{constraint2:BI}, for all $(i,j)\leftrightarrow(i,j')\in \mathcal{D}$, we have $P_{ij}\leq \hat{T}_{ij'}$ if $\hat{T}_{ij}\leq \hat{T}_{ij'}$. Then, at $t={P}_{ij'}$, we have $y_j(t,\hat{u}_j,x_j(0))=\beta_{ij}$ while $y_{j'}(t,\hat{u}_{j'},x_{j'}(0))\leq\alpha_{i'j}$. This implies that any two vehicles never meet inside a conflict area, that is, $\mathbf{y}(t,\hat{\mathbf{u}},\mathbf{x}(0))\notin B$ for all $t\geq 0$. 

Therefore, there exists $\hat{\mathbf{u}}$ such that $\mathbf{y}(t,\hat{\mathbf{u}},\mathbf{x}(0))\notin B$ for all $t\geq 0$.
\end{proof}

By Theorem~\ref{theorem:P1_equivalent_P3}, we can solve Problem~\ref{problem:verification} by solving Problem~\ref{problem:job-shop}. One may notice that Problem~\ref{problem:job-shop} is similar to the classical job-shop scheduling problem (Problem~\ref{problem:classical_jobshop}) if $D_{ij}=d_j-p_j$ and $P_{ij}=t_{ij}+p_j$. However, in Problem~\ref{problem:job-shop}, the release times, deadlines, and process times are defined for each operation as functions of the schedules. The fact that they vary depending on the schedules significantly complicates the problem. We thus cannot directly employ the solutions from the scheduling literature. Instead, we have to formulate a mixed-integer linear programming problem, which is proved to yield the equivalent answers to Problem~\ref{problem:job-shop} by assuming that the vehicle dynamics are single integrator dynamics.

\subsubsection{Mixed-integer programming}
Problem~\ref{problem:job-shop} can be transformed into a mixed-integer programming problem, which is an optimization problem subject to equality and inequality constraints in the presence of continuous and discrete variables. Notice that Constraint~\eqref{constraint1:physical} is already an inequality constraint. However, Constraint~\eqref{constraint2:BI} contains a disjunctive constraint, which can be converted into linear inequalities by introducing a binary variable $k_{ijj'}\in\{0,1\}$ and using the big-$M$ method \cite{grossmann_generalized_2012}. In particular, define $$k_{ijj'}:=\begin{cases}
1 & \begin{split}
\text{if vehicle}~j~\text{crosses conflict}&~\text{area}~i\\&\text{before vehicle}~j',
\end{split}\\ 
0 & \text{otherwise}.
\end{cases}$$
Also, let $M$ be a large positive constant in $\mathbb{R}$. Then Constraint~\eqref{constraint2:BI} can be rewritten as follows:
\begin{align}
\begin{split}\label{constraint:BI_mip}
\text{for all}~&(i,j)\leftrightarrow (i,j') \in\mathcal{D},\\
&k_{ijj'}+k_{ij'j}=1, ~~k_{ijj'},k_{ij'j}\in\{0,1\}\\
&P_{ij} \leq T_{ij'} + M(1-k_{ijj'}), \\
&P_{ij'}\leq T_{ij}+M(1-k_{ij'j}),
\end{split}
\end{align}
for $M$ sufficiently larger than $T_{ij}$ and $P_{ij}$ for all $(i,j)\in\mathcal{N}$. If $k_{ijj'}=1$ and $k_{ij'j}=0$, vehicle $j$ crosses conflict area $i$ before vehicle $j'$ so that $T_{ij}\leq T_{ij'}$. Then, $P_{ij}\leq T_{ij'}$ is imposed while $P_{ij'}\leq T_{ij}+M$ is automatically satisfied because of a sufficiently large $M$. Thus, \eqref{constraint:BI_mip} encodes the same constraint as \eqref{constraint2:BI}.

Notice that because $R_{ij}, D_{ij},$ and $P_{ij}$ are functions of variable $T_{ij}$, Problem~\ref{problem:job-shop} with Constraint~\eqref{constraint1:physical} and \eqref{constraint:BI_mip} is a general mixed-integer program (MIP). Due to its high complexity, this formulation is usually difficult to solve \cite{bussieck_minlp_2010}. If the constraints can be expressed in a linear function of variables, the problem becomes a mixed-integer linear program (MILP). Although MILP are combinatorial, several algorithmic approaches are available to solve medium to large size application problems \cite{floudas1995nonlinear}.

To this end, we assume that the longitudinal dynamics of vehicles are modeled as a single integrator as follows. For vehicle $j$,
\begin{align}\label{equation:single_integrator}
\dot{x}_j=u_j, && y_j=x_j.
\end{align}
Notice that the dynamic state $x_j\in X_j\subseteq \mathbb{R}$ is the position, and the control input $u_j\in U_j$ is the speed. Since vehicles do not go in reverse, we let $u_{j,min}>0$.  

With the first order dynamic model \eqref{equation:single_integrator}, we can transform Problem~\ref{problem:job-shop} into a mixed-integer linear programming problem.
Let us write $P_{ij}=T_{ij}+\min_{u_j} \{t:y_j(t,u_j,\alpha_{ij})=\beta_{ij}\}$ so that the constraint that $y_j(T_{ij},u_j,x_j(0))=\alpha_{ij}$ is automatically satisfied. By defining $$p_{ij}:=\{t:y_j(t,u_j,\alpha_{ij})=\beta_{ij}\},$$ $p_{ij}$ corresponds to the time spent inside conflict area $i$, independent of $T_{ij}$. Then, the variables for the mixed-integer linear programming problem are as follows:
\begin{itemize}
	\item $T_{ij}$ for $(i,j)\in\mathcal{N}$, continuous variables,
	\item $p_{ij}$ for $(i,j)\notin \mathcal{L}$, continuous variables,
	\item $k_{ijj'}$ and $k_{ij'j}$ for $(i,j)\leftrightarrow(i,j')\in\mathcal{D}$, binary variables.
\end{itemize}
Notice that $p_{ij}$ for $(i,j)\in \mathcal{L}$ is excluded from the variables because we can set $p_{ij}=(\beta_{ij}-\alpha_{ij})/u_{max}$. This is possible because $P_{ij}=T_{ij}+(\beta_{ij}-\alpha_{ij})/u_{max}$ by definition~\eqref{definition:P_in_L}, and  the minimum $p_{ij}$ is most likely to satisfy the problem formulated in the following paragraph.

Given the single integrator dynamics, we formulate the mixed-integer linear programming problem as follows.

\begin{problem}\label{problem:milp}
	Given $\mathbf{x}(0)$, determine if there exists a feasible solution subject to the following constraints.
	
\begin{enumerate}[A.]
	\item If $(i,j)\in \mathcal{F}$, for $y_j(0)<\alpha_{ij}$
	$$\frac{\alpha_{ij}-y_j(0)}{u_{j,max}}\leq T_{ij}\leq \frac{\alpha_{ij}-y_j(0)}{u_{j,min}}.$$
	\label{constraint1:firstoperation}
	For $\alpha_{ij}\leq y_j(0)$, consider $T_{ij}=0$.
	\item If $(i,j)\notin \mathcal{F}$, that is $\exists (i',j)$ such that $(i',j)\rightarrow (i,j)\in\mathcal{C}$,
	$$T_{i'j}+p_{i'j}+\frac{\alpha_{ij}-\beta_{i'j}}{u_{j,max}}\leq T_{ij}\leq T_{i'j}+p_{i'j}+\frac{\alpha_{ij}-\beta_{i'j}}{u_{j,min}}.$$
	\label{constraint2:sequential}
	
	\item If $(i,j)\notin \mathcal{L}$, for $y_j(0)<\alpha_{ij}$,
	$$\frac{\beta_{ij}-\alpha_{ij}}{u_{j,max}}\leq p_{ij}\leq \frac{\beta_{ij}-\alpha_{ij}}{u_{j,min}}.$$
	\label{constraint3:process}
	For $\alpha_{ij}\leq y_j(0)$, consider instead $\frac{\beta_{ij}-y_j(0)}{u_{j,max}}\leq p_{ij}\leq \frac{\beta_{ij}-y_{j}(0)}{u_{j,min}}$. If $\beta_{ij}\leq y_j(0)$, the schedule of operation $(i,j)$ is not of interest.
	
	\item For all $(i,j)\leftrightarrow (i,j')\in\mathcal{D}$, with a large number $M\in\mathbb{R}_+$,
	\begin{align*}
	&T_{ij}+p_{ij}\leq T_{ij'}+M(1-k_{ijj'}),\\
	&T_{ij'}+p_{ij'}\leq T_{ij} + M(1-k_{ij'j}),\\
	&k_{ijj'}+k_{ij'j} = 1.
	\end{align*}
	\label{constraint4:disjunctive}
\end{enumerate}
\end{problem}

We now prove that this problem yields equivalent answers to the job-shop scheduling problem (Problem~\ref{problem:job-shop}) with the first-order dynamics.

\begin{theorem}\label{theorem:P3_equivalent_P4}
	If the vehicle dynamics \eqref{equation:dynamics} are modeled as \eqref{equation:single_integrator}, Problem~\ref{problem:job-shop} is equivalent to Problem~\ref{problem:milp}.
\end{theorem}
\begin{proof}
	Problem~\ref{problem:job-shop} and Problem~\ref{problem:milp} have an identical instance $I=\{\mathbf{x}(0),\Theta\}$. Thus, we need to show that $I\in$~Problem~\ref{problem:job-shop} if and only if $I\in$~Problem~\ref{problem:milp}. We will prove that $I\in$~Problem~\ref{problem:job-shop} if $I\in$~Problem~\ref{problem:milp}, and $I\notin$~Problem~\ref{problem:job-shop} if $I\notin$~Problem~\ref{problem:milp}. 
	
	Suppose $I\in$~Problem~\ref{problem:milp}. Then there exist a feasible solution $(\tilde{\mathbf{T}},\tilde{\mathbf{p}},\tilde{\mathbf{k}})$ where $\tilde{\mathbf{T}}=\{\tilde{T}_{ij}:\forall(i,j)\in\mathcal{N}\}$, $\tilde{\mathbf{p}}=\{\tilde{p}_{ij}:\forall(i,j)\notin\mathcal{L}\}$, and $\tilde{\mathbf{k}}=\{\tilde{k}_{ijj'},\tilde{k}_{ij'j}:\forall(i,j)\leftrightarrow (i,j')\in\mathcal{D}\}$. 
	
	For $(i,j)\in\mathcal{F}$, $R_{ij}=(\alpha_{ij}-y_j(0))/u_{j,max}$ and $D_{ij}=(\alpha_{ij}-y_j(0))/u_{j,min}$ by definition \eqref{definition:RD_in_F}. For $(i,j)\notin\mathcal{F}$, that is, $\exists (i',j)\rightarrow (i,j)\in\mathcal{C}$, there is the constraint in definition~\eqref{definition:RD_notin_F} that $y_j(\tilde{T}_{i'j})=\alpha_{i'j}$. Thus, $R_{ij}$ and $D_{ij}$ are as follows.
	\begin{align*}
		R_{ij}=\tilde{T}_{i'j}+\frac{\alpha_{ij}-\alpha_{i'j}}{u_{j,max}}=\tilde{T}_{i'j}+\tilde{p}_{i'j}+\frac{\alpha_{ij}-\beta_{i'j}}{u_{j,max}}, \\	D_{ij}=\tilde{T}_{i'j}+\frac{\alpha_{ij}-\alpha_{i'j}}{u_{j,min}}=\tilde{T}_{i'j}+\tilde{p}_{i'j}+\frac{\alpha_{ij}-\beta_{i'j}}{u_{j,min}}.
	\end{align*}
	The second equalities in both equations result from Constraint \ref{constraint3:process}. Therefore, Constrains \ref{constraint1:firstoperation} and \ref{constraint2:sequential} imply $R_{ij}\leq \tilde{T}_{ij}\leq D_{ij}$ for all $(i,j)\in\mathcal{N}$ (Constraint \eqref{constraint1:physical}).
		
	In Constraint \ref{constraint4:disjunctive}, we have $P_{ij}\leq \tilde{T}_{ij}+\tilde{p}_{ij}$ because $P_{ij}$ is the minimum time to reach $\beta_{ij}$. Therefore, we have $P_{ij}\leq \tilde{T}_{ij}+\tilde{p}_{ij}\leq \tilde{T}_{ij'}+M(1-\tilde{k}_{ijj'})$. Similarly, $P_{ij'}\leq \tilde{T}_{ij'}+\tilde{p}_{ij'}\leq \tilde{T}_{ij}+M(1-\tilde{k}_{ij'j})$ (Constraint~\eqref{constraint:BI_mip}).
	
	Thus, $\tilde{\mathbf{T}}$ satisfies the constraints in Problem~\ref{problem:job-shop}. That is, $I\in$ Problem~\ref{problem:job-shop}.

	Suppose $I\notin$~Problem~\ref{problem:milp}. Notice that if Constraint \ref{constraint3:process} is ignored and let $p_{ij}=0$, the problem is always feasible because for $(i_1,j)\in\mathcal{F}$ and $(i_1,j)\rightarrow (i_2,j),\ldots, (i_{d-1},j)\rightarrow (i_d,j)\in\mathcal{C}$,
	\begin{align*}
	&T_{i_1,j}=\frac{\alpha_{i_1j}-y_j(0)}{u_{j,max}},~T_{i_2j}=T_{i_1j}+\frac{\alpha_{i_2j}-\beta_{i_1j}}{u_{j,max}},\\
	&\hspace{1 in}\ldots,~ T_{i_dj}=T_{i_{d-1}j}+\frac{\alpha_{i_dj}-\beta_{i_{d-1}j}}{u_{j,max}}
	\end{align*} becomes a feasible solution for any $j$. Constraint~\ref{constraint4:disjunctive} is also satisfied because either $T_{ij}\leq T_{i'j}$ or $T_{i'j}\leq T_{ij}$ is always true. We can thus find the maximum process time that is a feasible solution for the problem without Constraint~\ref{constraint3:process}. Since $I\notin$~Problem~\ref{problem:milp}, this solution violates Constraint~\ref{constraint3:process}. Thus, there is no $p_{ij}\geq (\beta_{ij}-\alpha_{ij})/{u_{j,max}}$ for any $(i,j)\notin\mathcal{L}$ such that Constraints \ref{constraint1:firstoperation}, \ref{constraint2:sequential}, and \ref{constraint4:disjunctive} are satisfied. This, in turn, implies that given the definition $P_{ij}=T_{ij}+\min p_{ij}$, there is no $P_{ij}\geq T_{ij}+(\beta_{ij}-\alpha_{ij})/u_{j,max}$ such that the constraints in Problem~\ref{problem:job-shop} are satisfied. Since $P_{ij}$ is not feasible, neither are $T_{ij}$ and $k_{ijj'}$. Thus, $I\notin$ Problem~\ref{problem:job-shop}.
\end{proof}

We solve Problem~\ref{problem:milp} using CPLEX. The procedure that solves Problem~\ref{problem:milp} given an instance $I=\{\mathbf{x}(0),\Theta\}$ is referred to as \texttt{Jobshop}$(I)$. If $I\in$ Problem~\ref{problem:milp}, that is, $I\in$ Problem~\ref{problem:verification} by Theorems~\ref{theorem:P1_equivalent_P3} and \ref{theorem:P3_equivalent_P4}, $\mathtt{Jobshop}(I)$ returns $\{\mathbf{T,p},yes\}$. Otherwise, it returns $\{\emptyset,\emptyset,no\}$. 

\subsection{Solution of Problem~\ref{problem:supervisor}}\label{section:supervisor}
 The supervisor runs in discrete time with a time step $\tau$. At time $k\tau$ where $k>0$, it receives the measurements of the states $\mathbf{x}(k\tau)$ and drivers' inputs $\mathbf{u}_{driver}^k\in\mathbf{U}$ of the vehicles near an intersection. By assuming that $\mathbf{u}_{driver}^k$ is constant for time $[k\tau, (k+1)\tau)$, we predict a state at the next time step, called a state prediction and denoted by $\hat{\mathbf{x}}(\mathbf{u}_{driver}^k)$, as follows.
$$\hat{\mathbf{x}}(\mathbf{u}_{driver}^k) = \mathbf{x}(\tau, \mathbf{u}_{driver}^k, \mathbf{x}(k\tau)).$$
 
 Notice that $\mathtt{Jobshop}(\hat{\mathbf{x}}(\mathbf{u}_{driver}^k),\Theta)$ determines whether or not collisions can be avoided at all future time given the state prediction. If it returns $\{\mathbf{T,p},yes\}$, then the supervisor allows the vehicles to drive with input $\mathbf{u}_{driver}^k$ for time $[k\tau, (k+1)\tau)$. The schedule $\mathbf{T}$ and the process time $\mathbf{p}$ are used to generate a safe input signal $\mathbf{u}_{safe}^{k+1,\infty}$, defined on time $[(k+1)\tau,\infty)$. We define a safe input operator $\sigma(\hat{\mathbf{x}}(\mathbf{u}_{driver}^k),\mathbf{T,p})$ as follows.
 \begin{align}
 \begin{split}\label{equation:safe_input}
 &\sigma(\hat{\mathbf{x}}(\mathbf{u}_{driver}^k),\mathbf{T,p})\\
 &\in \{(u_1,\ldots,u_n)\in\mathcal{U}: y_j(T_{ij},u_j,\hat{x}_j(u_{driver,j}^k))=\alpha_{ij}\\
 &\hspace{0.9 in}\text{and}~y_j(p_{ij},u_j,\alpha_{ij})=\beta_{ij}~\forall (i,j)\in\mathcal{N}\},
 \end{split}
 \end{align}
 where $u_{driver,j}^k$ is the $j^{th}$ entry of $\mathbf{u}_{driver}^k$, and $\hat{x}_j(u_{driver,j}^k)$ is the $j^{th}$ entry of $\hat{\mathbf{x}}(\mathbf{u}_{driver}^k)$. This safe input signal is stored for possible uses at the next time step.
 
 If $\mathtt{Jobshop}(\hat{\mathbf{x}}(\mathbf{u}_{driver}^k),\Theta)$ returns $\{\emptyset,\emptyset,no\}$, then the supervisor overrides the vehicles using the safe input signal stored at the previous step, $\mathbf{u}_{safe}^{k,\infty}$. Since $\mathbf{u}_{safe}^{k,\infty}$ is defined on time $[k\tau, \infty)$, let $\mathbf{u}_{safe}^{k}\in\mathcal{U}$ be $\mathbf{u}_{safe}^{k,\infty}$ restricted to time $[k\tau, (k+1)\tau)$. The supervisor blocks the drivers' inputs $\mathbf{u}_{driver}^k$ and returns the safe input $\mathbf{u}_{safe}^k$ for time $[k\tau, (k+1)\tau)$ to prevent future collisions.

This procedure is written as an algorithm as follows.
\begin{algorithm}[H]
	\caption{$\mathtt{Supervisor}(\mathbf{x}(k\tau),\mathbf{u}_{driver}^k)$}
	\label{algorithm:supervisor}
	\begin{algorithmic}[1]
		\State $\{\mathbf{T}_1,\mathbf{p}_1,answer_1\}=\mathtt{Jobshop}(\hat{\mathbf{x}}(\mathbf{u}_{driver}^k),\Theta)$\label{algorithm:solve_verification1}
		\If{$answer_1=yes$}
		\State $\mathbf{u}^{k+1,\infty}\gets \sigma(\hat{\mathbf{x}}(\mathbf{u}_{driver}^k),\mathbf{T}_1,\mathbf{p}_1)$\label{algorithm:1_safe_input}
		\State $\mathbf{u}^{k+1}_{safe}\gets \mathbf{u}^{k+1,\infty}(t)$ for $t\in [(k+1)\tau, (k+2)\tau)$\label{algorithm:1_safe_input2}
		\State \textbf{return} $\mathbf{u}_{driver}^k$\label{algorithm:doesnot_intervene}
		\Else
		\State $\{\mathbf{T}_2,\mathbf{p}_2,answer_2\} = \mathtt{Jobshop}(\hat{\mathbf{x}}(\mathbf{u}_{safe}^k),\Theta)$\label{algorithm:verification_again}
		\State $\mathbf{u}^{k+1,\infty}\gets \sigma(\hat{\mathbf{x}}(\mathbf{u}_{safe}^k),\mathbf{T}_2,\mathbf{p}_2)$\label{algorithm:2_safe_input}
		\State $\mathbf{u}_{safe}^{k+1}\gets \mathbf{u}^{k+1,\infty}(t)~ \text{for}~ t\in[(k+1)\tau, (k+2)\tau)$\label{algorithm:2_safe_input2}
		\State \textbf{return} $\mathbf{u}_{safe}^k$\label{algorithm:override}
		\EndIf
	\end{algorithmic}
\end{algorithm}
If $answer_1=yes$, then the supervisor generates and stores the safe input $\mathbf{u}_{safe}^{k+1}$ in lines~\ref{algorithm:1_safe_input}-\ref{algorithm:1_safe_input2}, and does not intervene in line~\ref{algorithm:doesnot_intervene}. If $answer_1=no$, the supervisor solves the verification problem in line~\ref{algorithm:verification_again} given the state predicted with the safe input $\mathbf{u}_{safe}^k$. It will be proved in the following theorem that $answer_2$ is always $yes$, which implies the non-blocking property of the supervisor. Based on $\mathbf{T}_2$ and $\mathbf{p}_2$, the supervisor generates and stores the safe input $\mathbf{u}_{safe}^{k+1}$ in lines~\ref{algorithm:2_safe_input}-\ref{algorithm:2_safe_input2}, and overrides the vehicles in line~\ref{algorithm:override}.

\begin{theorem}
	Algorithm~\ref{algorithm:supervisor} solves Problem~\ref{problem:supervisor}.
\end{theorem}
\begin{proof}
	To prove that Algorithm~\ref{algorithm:supervisor} is a solution of Problem~\ref{problem:supervisor}, we check if the algorithm satisfies the specifications in Problem~\ref{problem:supervisor}.
	
	Specification~\ref{specification1} is met by the design of the algorithm. If there exists $\mathbf{u}\in\mathcal{U}$ such that $\mathbf{y}(t,\mathbf{u},\hat{\mathbf{x}}(\mathbf{u}_{driver}^k))\notin B$ for all $t\geq 0$, then $\mathtt{Jobshop}(\hat{\mathbf{x}}(\mathbf{u}_{driver}^k),\Theta)$ returns \textit{yes}. In this case, the supervisor returns $\mathbf{u}_{driver}^k$. Otherwise, it returns $\mathbf{u}_{safe}^k\in\mathcal{U}$. The fact that this input makes $\mathtt{Jobshop}(\hat{\mathbf{x}}(\mathbf{u}_{safe}^k),\Theta)$ return \textit{yes} will be clear in the proof of the non-blocking property. 
	
	To prove the non-blocking property, we use mathematical induction on $k$ where $t=k\tau$. At $t=0$, we assume $\mathbf{u}_{safe}^0 \neq \emptyset$. At $t=(k-1)\tau$, suppose there exists $\mathbf{u}_{safe}^{k-1}$. That is, by definition, there exists $\mathbf{u}'\in\mathcal{U}$ such that $\mathbf{y}(t,\mathbf{u}',\hat{\mathbf{x}}(\mathbf{u}_{safe}^{k-1}))\notin B$ for all $t\geq 0$. If $\mathtt{Jobshop}(\hat{\mathbf{x}}(\mathbf{u}_{driver}^{k-1}),\Theta)$ returns \textit{yes}, then then there exists $\mathbf{u}\in\mathcal{U}$ such that $\mathbf{y}(t,\mathbf{u},\hat{\mathbf{x}}(\mathbf{u}_{driver}^{k-1}))\notin B$ for all $t\geq 0$ by Problem~\ref{problem:verification}. 
	
	Now at $t=k\tau$, we want to prove that there exists $\mathbf{u}_{safe}^k$. Notice that $\mathbf{x}(k\tau)$ is either $\hat{\mathbf{x}}(\mathbf{u}_{driver}^{k-1})$ or $\hat{\mathbf{x}}(\mathbf{u}_{safe}^{k-1})$. In the former case, let $\mathbf{u}^{k}$ be $\mathbf{u}$ restricted to time $[k\tau, (k+1)\tau)$, and $\mathbf{u}^{k+1,\infty}$ be $\mathbf{u}$ restricted to time $[(k+1)\tau, \infty)$. Then, we have $$\mathbf{y}(t,\mathbf{u}^{k+1,\infty},{\mathbf{x}}(\tau,\mathbf{u}^k,\hat{\mathbf{x}}(\mathbf{u}_{driver}^{k-1})))\notin B.$$ Thus there exists $\mathbf{u}_{safe}^k = \mathbf{u}^k$. Similarly for the latter case, let $\mathbf{u}'^k$ be $\mathbf{u}'$ restricted to time $[k\tau, (k+1)\tau)$, and $\mathbf{u}'^{k+1,\infty}$ be $\mathbf{u}'$ restricted to time $[(k+1)\tau,\infty)$. Then, we have
	$$\mathbf{y}(t,\mathbf{u}'^{k+1,\infty},\mathbf{x}(\tau,\mathbf{u}'^k,\hat{\mathbf{x}}(\mathbf{u}_{safe}^{k-1})))\notin B.$$ Thus there exists $\mathbf{u}_{safe}^k = \mathbf{u}'^k$. Therefore, in any case, there exists a safe input $\mathbf{u}_{safe}^k$.
	
	If $\mathbf{u}_{safe}^0$ exists, there exists $\mathbf{u}_{safe}^k$ for any $k>0$. The supervisor is thus, non-blocking.
	\end{proof}
	\begin{figure*}[htb!]
	\centering
	\subfigure[Bad set]{	
		\includegraphics[width=0.99\columnwidth]{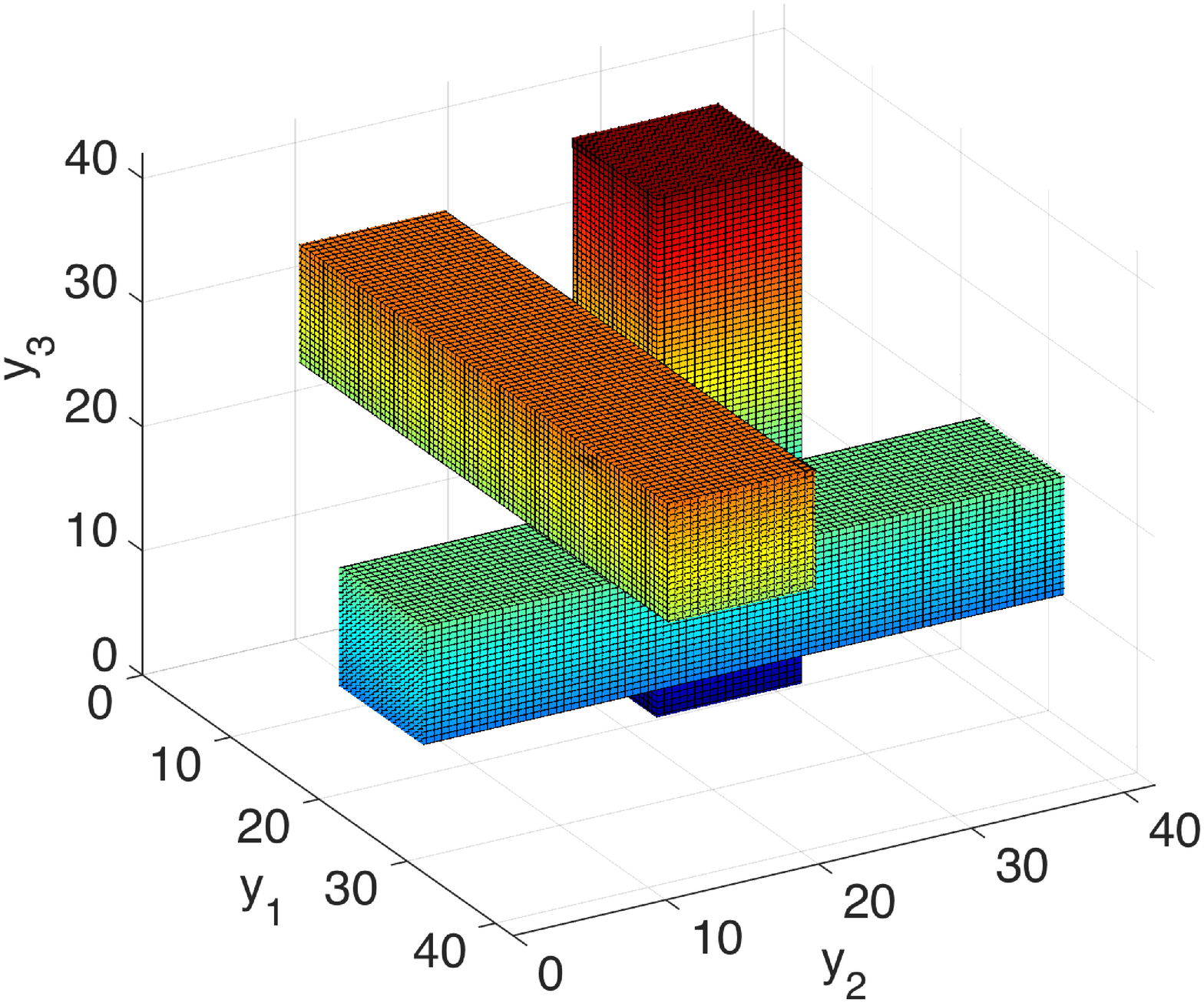}
		\label{figure:badset}}
	\quad
	\subfigure[Capture set]{	
		\includegraphics[width=0.99\columnwidth]{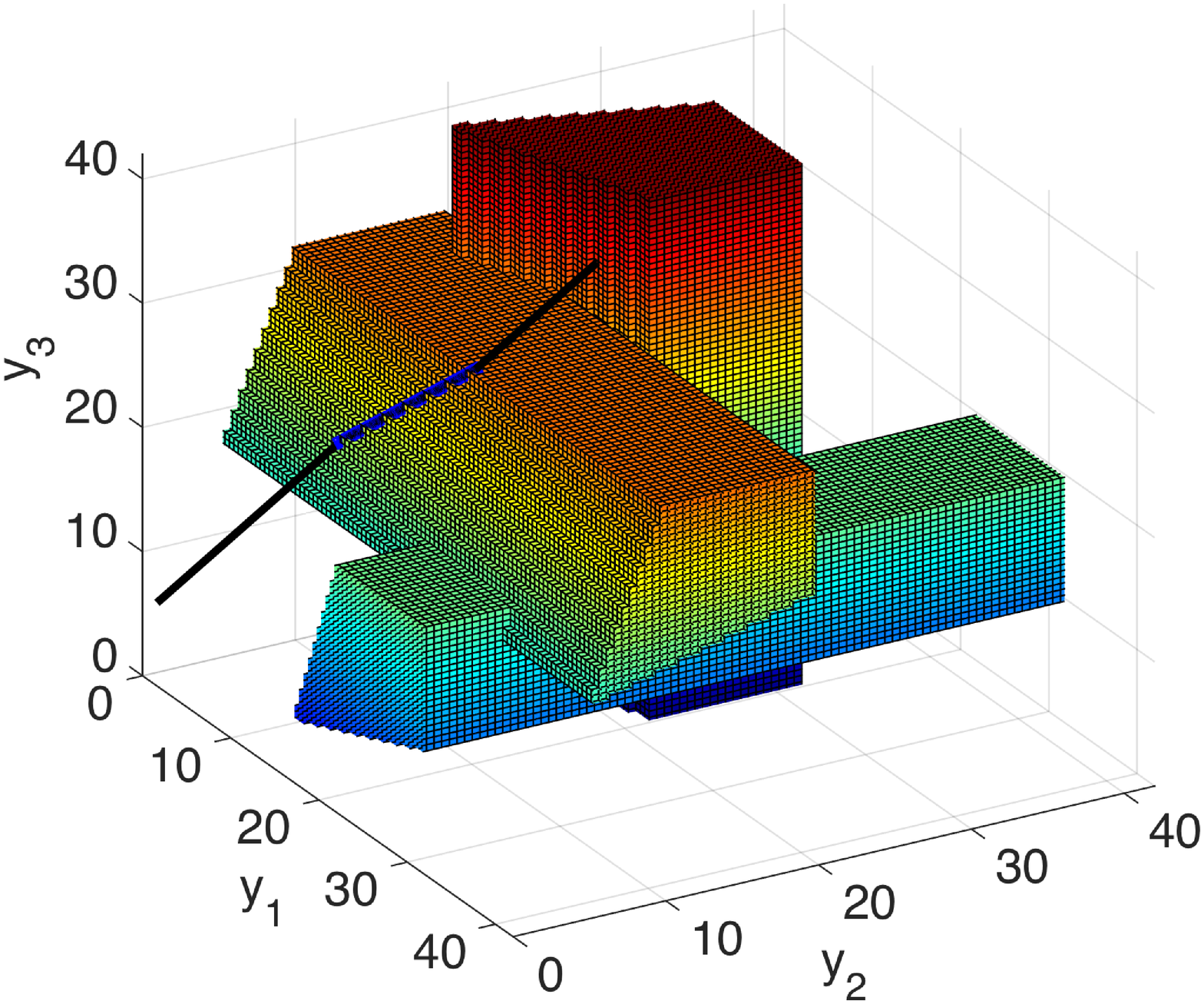}
		\label{figure:captureset}}
	
	\caption{Position space of the three vehicles in the scenario of Figure~\ref{figure:scenario}. Subfigure (a) shows the bad set defined in \eqref{equation:badset}, and subfigure (b) shows the resulting capture set defined in \eqref{equation:capture_Set}. In (b), the black line is the trajectory of the system, and the blue thick line highlights the positions at times when the supervisor overrides the vehicles. Notice that the supervisor prevents them from entering the capture set, thereby averting collision.}
	\label{figure:3D}
	\end{figure*}

\section{Simulation Results}\label{section:simulation_results}
This section presents simulation results of the supervisor. In particular, considering the intersection scenarios illustrated in Figures~\ref{figure:general_intersection} and \ref{figure:scenario}, we validate that the supervisor prevents impending collisions by minimally overridng vehicles. Also, the simulations illustrate that for a system with a large number of vehicles, the computation time required for the supervisor algorithm (Algorithm~\ref{algorithm:supervisor}) at each step is within the allotted 100 ms.

 We implement Algorithm~\ref{algorithm:supervisor} using MATLAB, in which mixed-integer programming in Problem~\ref{problem:milp} is solved by using CPLEX. To speed up the process of generating the constraints of the problem, MATLAB Coder\texttrademark \cite{Matlab_coder_2015} is used to replace the code written in MATLAB with the C code and compile it into a MATLAB executable function. Simulations are performed on a personal computer, which runs Windows 7 Home Premium and consists of an Intel Core i7-3770s processor at 3.10 GHz and 8 GB random-access memory.

Consider first Figure~\ref{figure:scenario}, in which three vehicles are approaching the intersection containing three conflict points. The parameters used in the simulations are $\tau=$0.1, $U_j=[0.1,0.3]$ for all $j\in\{1,\ldots,n\}$, $(\alpha_{ij},\beta_{ij})=$(10,20) for $(i,j)\in\mathcal{F}$, and $(\alpha_{ij}, \beta_{ij})=(\alpha_{i'j} + 22,\alpha_{ij} + 10)$ for $(i,j)\notin\mathcal{F}$, where $(i',j)\rightarrow(i,j)\in\mathcal{C}$. 

To solve the verification problem (Problem~\ref{problem:verification}), the work in \cite{hafner_cooperative_2013} considers the set of initial states such that no input exists to avoid a collision. This subset of the state space is called the \textit{capture set} and defined as follows.
	\begin{align}\label{equation:capture_Set}
	\mathcal{CS}:=\{\mathbf{x}\in \mathbf{X}: \forall \mathbf{u}\in\mathcal{U},~\exists t~\text{such that}~\mathbf{y}(t,\mathbf{u},\mathbf{x})\in B\}.
	\end{align}
	The capture set resulting from the bad set in Figure~\ref{figure:badset} is shown in Figure~\ref{figure:captureset}.
Given an instance $I=\{\mathbf{x}(0),\Theta\}$ of Problem~\ref{problem:verification}, $I\notin$~Problem~\ref{problem:verification} if and only if $\mathbf{x}(0)\in \mathcal{CS}$ by definition. By Theorems~\ref{theorem:P1_equivalent_P3} and \ref{theorem:P3_equivalent_P4}, if $\mathbf{x}(0)\in\mathcal{CS}$, $I\notin$~Problems~\ref{problem:job-shop} and \ref{problem:milp}.

In Figure~\ref{figure:captureset}, the black line represents the trajectory of the system given an initial condition $\mathbf{x}(0)=$(-2.8,-3.7,-1.2). When the supervisor overrides the vehicles, the trajectory is shown in blue. The drivers' inputs are set to be $\mathbf{u}_{driver}^k=(0.15, 0.11, 0.25)$ and constant for all $k\geq 0$ where $t=k\tau$, so that without override actions of the supervisor, the trajectory would enter the bad set in Figure~\ref{figure:badset}. Notice that the supervisor overrides the vehicles right before the trajectory enters the capture set and makes the trajectory ride on the boundary of the capture set. The drivers regain the control of their vehicles once the dangerous situation is resolved. This confirms that the supervisor is least restrictive because it intervenes only when the state prediction $\hat{\mathbf{x}}(\mathbf{u}_{driver}^k)$ enters the capture set. The computation of the supervisor algorithm (Algorithm~\ref{algorithm:supervisor}) takes less than 4 ms per iteration in the worst case.

\begin{figure}[htb!]
	\includegraphics[width = \columnwidth]{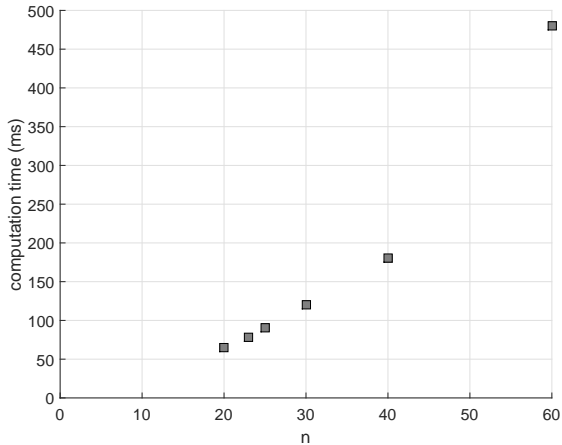}
	\caption{Computation time for one iteration of  Algorithm~\ref{algorithm:supervisor} in the worst case with respect to the number of vehicles. }
	\label{figure:complexity}
\end{figure}
We then run Algorithm~\ref{algorithm:supervisor} for the intersection instance shown in Figure~\ref{figure:general_intersection}, which contains twenty vehicles and forty eight conflict points. Then, we inserted additional vehicles per road (far enough so to ensure that rear-end collsions do not occur) to determine how many vehicles the supervisor can handle within the 100 ms. In Figure~\ref{figure:complexity}, the computation time required for one iteration of Algorithm~\ref{algorithm:supervisor} is shown with respect to the number of vehicles. Notice that as the number of vehicles increases, the computation time increases exponentially. Although the problem is not scalable, about twenty five vehicles can be managed by the supervisor within the 100 ms even in the complicated intersection scenario. 

The intersection scenario of Figure~\ref{figure:general_intersection} is created from the top 20 crash intersection locations in the report of the Massachusetts Department of Transportation \cite{MassDOT_2012_Topcrash} such that it can represent each intersection topology by removing or combining its lanes. That is, this intersection scenario consisting of twenty lanes and forty eight conflict points is more complicated than the twenty most dangerous intersections in Massachusetts. If we do not consider rear-end collisions and assume that there is only one vehicle per road, the number of vehicles in typical intersection scenarios usually does not exceed twenty. We can thus conclude that this supervisor is practical for typical intersection scenarios. How accounting for rear-end collisions affects computational complexity will be investigated in future work. It is shown in \cite{colombo_least_2014} that additional vehicles on the same lane increase computational complexity less than those on different lanes due to precedence constraints. Since in Figure~\ref{figure:complexity}, we did not consider these precedence constraints, we expect that the computation time will be lower than that shown in Figure~\ref{figure:complexity}.

\section{Conclusions}\label{section:conclusions}
We have designed a supervisor that overrides human-driven vehicles only when a future collision is detected and has a non-blocking property. To this end, we have formulated the verification problem and the job-shop scheduling problem and proved that they are equivalent. To solve the job-shop scheduling problem, we have converted it into a mixed-integer linear programming problem by assuming the single integrator vehicle dynamics. The computer simulations confirm that the supervisor guarantees safety while overriding vehicles only when a future collision is unavoidable otherwise. Also, the computational studies show that despite the combinatorial complexity of the verification problem, the supervisor can deal with a complicated intersection scenario as in Figure~\ref{figure:general_intersection} within the allotted 100 ms per iteration. 

While this paper considers a general intersection model in terms of conflict areas, the inclusion of rear-end collisions in the scenario makes it more practical. Moreover, to account for more realistic dynamic behaviors of vehicles, a nonlinear second-order model will be considered. In particular, for second-order linear dynamics, the job-shop scheduling problem may be reformulated as a mixed-integer quadratic programming problem. Also, as considered in \cite{bruni_robust_2013,ahn_supervisory_2014,ahn_experimental_2015} in which an intersection is modeled as a single conflict area, measurement and process uncertainty and the presence of unequipped vehicles will be investigated in future work. Undetermined routes of vehicles will also be investigated by considering possible decisions of steering inputs.

\section{Acknowledgments}
The authors would like to thank Alessandro Colombo and Gabriel Campos at Politecnico di Milano for the helpful discussions and suggestions. This work was in part supported by NSF CPS Award No. 1239182.

\bibliographystyle{abbrv}
\bibliography{IEEEabrv} 

%

\end{document}